%% file: errprob.tex
\documentclass{article}

\usepackage[a4paper]{geometry}
\usepackage{url,amsmath,amssymb,amsthm,enumerate,graphicx,pdflscape}
\usepackage{mathtools}
\usepackage{algorithm}
\usepackage{algorithmicx}
\usepackage[noend]{algpseudocode}
\usepackage{rotating}

\usepackage{xcolor,tikz}
\usetikzlibrary{shapes}

\usepackage{thm-restate}

\newtheorem{proposition}{Proposition}
\newtheorem{corollary}{Corollary}

\newtheorem{lemma}{Lemma}

\newtheorem{definition}{Definition}
\newtheorem{remark}{Remark}
\newtheorem{assumption}{Assumption}

\newcommand{\F}[1]{\mathbb{F}_{\!#1}}

\newcommand{\cC}{\mathcal{C}}

\newcommand{\cD}{\mathcal{D}}

\newcommand{\ev}{{\mathbf{e}}}
\newcommand{\hv}{{\mathbf{h}}}

\newcommand{\pv}{{\mathbf{p}}}
\newcommand{\qv}{{\mathbf{q}}}

\newcommand{\yv}{{\mathbf{y}}}

\newcommand{\Hm}{{\mathbf{H}}}

\newcommand{\Imat}{{\mathbf{I}}}

\newcommand{\Pm}{{\mathbf{P}}}

\newcommand{\mat}[1]{\ensuremath{\boldsymbol{#1}}}
\newcommand{\un}{{\mat{1}}}
\newcommand{\zero}{{\mat{0}}}

\DeclarePairedDelimiterX\Kullc[2]{(}{)}{#1 \delimsize\vert\!\delimsize\vert #2}
\newcommand{\Kull}[2]{D\Kullc*{#1}{#2}}
\newcommand{\prob}{\mathbb{P}}

\newcommand{\BO}[1]{O\left(#1\right)}
\newcommand{\BOm}[1]{\Omega\left(#1\right)}
\newcommand{\eqdef}{\stackrel{\text{def}}{=}}
\newcommand{\Perr}{P_{\text{err}}}

\title{The decoding failure probability of MDPC codes}
\author{Jean-Pierre Tillich \footnote{Inria, SECRET Project, 2 Rue Simone Iff 75012 Paris Cedex, France, Email: \texttt{\{jean-pierre.tillich\}@inria.fr}. Part of this work was supported by the Commission of the European Communities through the Horizon 2020 program under project number 645622 PQCRYPTO.}}

\begin{document}
\maketitle

\begin{abstract}
Moderate Density Parity Check (MDPC) codes are defined here as codes which have a parity-check matrix whose row weight is $O(\sqrt{n})$ where $n$ is 
the length $n$ of the code. They can be decoded like LDPC codes but they decode much less errors than LDPC codes: the number of errors they can decode in this case is of order $\Omega(\sqrt{n})$.
Despite this fact they have been proved very useful in cryptography for devising key exchange mechanisms. They have also been proposed 
in  McEliece type cryptosystems. However in this case, the parameters that have been proposed in \cite{MTSB13} were broken in \cite{GJS16}. This attack exploits 
the fact that the decoding failure probability  is non-negligible. We show here that this attack can be thwarted by choosing the parameters in a more conservative way.
We first show that such codes can decode with a simple bit-flipping decoder any pattern of $O\left(\frac{\sqrt{n} \log \log n }{\log n}\right)$ errors.
This avoids the previous attack at the cost of significantly increasing the key size of the scheme. We then show that under a very reasonable assumption the decoding failure
 probability  decays almost exponentially with the codelength with just two iterations of bit-flipping. With an additional assumption it 
has even been proved that it decays exponentially with an unbounded number of iterations and we show that in this case the increase of the key size which is required 
for resisting to the  attack of \cite{GJS16} is only moderate.

\end{abstract}

\input{introduction}

\input{majoritylogic}

\input{analysis}
\input{numerical}
\input{conclusion}

\bibliographystyle{alpha}
\bibliography{codecrypto}

\appendix
\input{appendix}

\end{document}

%% file: introduction.tex
\section{Introduction}
Virtually all the public key cryptography used in practice today can be attacked in polynomial time by a quantum computer \cite{S94a}.
Even if such a quantum computer does not exist yet, finding viable solutions which would be resistant to a quantum computer is expected to be a lengthy process. This is one of the reasons why the
 NIST has recently launched a process for standardizing public key cryptographic algorithms that would be safe against a 
quantum adversary. 
 Code-based cryptography is believed to be quantum resistant and is 
 therefore considered as a viable solution. 
The McEliece system \cite{M78} based on binary Goppa codes, which is almost as old as RSA, is a public key cryptosystem that 
falls into this category. It has withstood all cryptanalyses up to now. 
It is  well known to provide extremely fast encryption and 
fast decryption \cite{BS08,BCS13},
 but has large public keys, about 200 kilobytes
for 128 bits of security and slightly less than one megabyte for 256 bits of security \cite{BLP08}.

There have been many attempts to decrease the key size of this system. One of the most satisfying answer up to now has been to use 
 Moderate Density Parity Check (MDPC) codes. The rows of their parity-check that defines them is of order 
 $\BO{\sqrt{n}}$ when $n$ is the length of the code. This family is very attractive
 since (i) the decryption algorithm is extremely simple and is based on a very simple bit flipping decoding algorithm, (ii) direct attacks on the key really amount to a problem of the same nature as decoding a linear code. This can be used to give a security proof \cite{MTSB13}.
 
 This work builds upon the following observation: 
 decoding $w$ errors in a generic linear code $\cC$ or finding a codeword  of weight $w$ in the linear code where 
 this word has been added to $\cC$ are problems that are polynomially equivalent. This problem is considered very hard when $t$ or $w$ are large enough: after decades of active research
\cite{P62,LB88,L88,S88,D91,B97b,FS09,BLP11,MMT11,BJMM12,MO15,DT17,BM17} the best algorithms solving this issue
 are still exponential in $t$ (or $w$), their complexity is of the form
 $\frac{e^{t \alpha(R)(1+o(1))}}{N}$  (replace $t$ by $w$ in the low-weight search problem) where $N$ is the number of solutions of the problem and $R$ is the rate of the code. This holds even for algorithms in the quantum computing model \cite{B11,KT17a}. Moreover, the relative exponent $\alpha(R)$ has decreased only very slowly after $50$ years of active research on the topic.The proposal made in \cite{MTSB13} exploits this. It suggests to use MDPC codes of rate $\frac{1}{2}$ of a certain length $n$ in a McEliece scheme which are able to decode $t = \Theta(\sqrt{n})$ errors with parity-check equations of weight $w \approx t$. Recovering 
 the plaintext without knowing the secret parity-check matrix of the code amounts to decode $\Theta(\sqrt{n})$ errors 
 in a linear code which is conjectured to be hard \cite{A11}, whereas recovering the secret MDPC structure amounts to
 find codewords of roughly the same weight as $t$ in the dual code that has the same dimension as the code that is decoded. Both problems are equally hard here. Note that this is not the case if we would have taken LDPC codes. They can decode much larger errors, however in this case finding the low weight parity-checks can be done in polynomial time which breaks the system completely as observed in \cite{MRS00}.
 
 However there is a problem with the security proof of \cite{MTSB13} because it does not take into account the decoding failure probability. This is not not necessarily a problem in a setting where the scheme is used to devise ephemeral keys \cite{BGGMPST17,ABBBBDGGGMMPSTZ17}.
 However, in security models where an attacker is allowed to query the decryption oracle many times, this can be a problem as observed by 
\cite{GJS16} which showed how to attack the parameters proposed in \cite{MTSB13}. 
 This attack really exploits the non negligible decoding failure probability of the MDPC codes chosen in \cite{MTSB13}. If this probability were as low as $2^{-\lambda}$ where $2^\lambda$ is the
complexity of the best attack that the scheme has to sustain, then this would not be a problem and the security proof of \cite{BGGMPST17} 
could be used to show the security of the scheme under this stronger attacking model. 
This raises the issue whether or not
the error probability of MDPC codes can be made extremely small for affordable parameters.

We tackle this issue by giving several different answers to this issue. We study in depth this question in the regime which is particularly interesting for these cryptographic applications, namely when the weight of the parity-check equations is of order $\BO{\sqrt{n}}$ where 
$n$ is the length of the MDPC code. We define in the whole article MDPC codes in this way
\begin{definition}[MDPC code]  An MDPC code is a binary linear code that admits a parity check matrix whose rows are all of weight $\BO{\sqrt{n}}$ where $n$ is the length of the code.  In the case where this parity-check matrix has rows of a same weight $w$ and columns of a same weight $v$, we say that the parity-check matrix is of type $(v,w)$. By some abuse of terminology, we will also call the corresponding code a code of type $(v,w)$.
\end{definition}

We will decode these codes with an even simpler bit-flipping decoding algorithm than the one considered in \cite{MTSB13} to simplify the analysis. 
One round of decoding is just majority-logic decoding based 
on a sparse parity-check matrix of the code. When we perform just one round of bit-flipping we call this decoder a majority-logic 
decoder. Recall that a majority logic decoder based on a certain parity-check matrix computes for all bits the number $u_i$ of 
parity-checks that involve the bit $i$ that are are unsatisfied. Let $n_i$ be the number of parity-checks involving bit $i$. 
If for a bit $i$ we have $u_i > n_i/2$ (i.e if a strict majority of such parity-checks is violated) the bit gets flipped. 
The parity-check equations used for making the decision on a given bit may depend on the bit (in particular they may have disjoint supports outside the bit they help to decode). This is not the path we follow here. It turns out that for an MDPC code, we can use all  the parity-check equations defining the MDPC code without too much penalty 
in doing so. We will assume here that the computation of the $v_i$'s is done in parallel so that flipping one bit does not affect the other $v_j$'s. In other words the decoder works as given in Algorithm \ref{algo:bf} when we perform 
$N$ iterations.

\begin{algorithm}
{\bf Input:} $\yv \in \{0,1\}^n$ the word that has to be decoded\\
{\bf Output:}  the decoded word.
\caption{Bit-flipping decoder}
\label{algo:bf}
\begin{algorithmic}
\ForAll{$i \in \{1,\dots,n\}$}
\State{$n_i \gets \#\left\{j \in \{1,\dots,r\}:h_{ji}=1\right\}$}
\EndFor
\For{$a=1$ to $N$}
\ForAll{$i \in \{1,\dots,n\}$}
\State{$u_i \gets \#\left\{j \in \{1,\dots,r\}:h_{ji}=1,\;\bigoplus_{\ell} h_{j\ell} y_{\ell}= 1\right\}$}
\EndFor
\ForAll{$i \in \{1,\dots,n\}$}
\If{$u_i >n_i/2$} \label{step:threshold}
\State{$y_i \gets 1-y_i$}
\EndIf
\EndFor
\EndFor
\Return{$\yv$}
\end{algorithmic}
\end{algorithm}

A crucial quantity will play an important role, namely

\begin{definition}[maximum column intersection] Let $\Hm=(h_{ij})_{\substack{1 \leq i \leq r \\1 \leq j \leq n}}$ be a binary matrix. The intersection number of two different columns $j$ and $j'$ of $\Hm$ is equal to the number of rows $i$ for which $h_{ij}=h_{ij'}=1$. 
 The 
{\em maximum column intersection} of $\Hm$ is equal to the maximum intersection number of two different columns of $\Hm$.
\end{definition}

The point is that it is readily verified (see Proposition \ref{prop:super_easy}) that an MDPC code of type $(v,w)$ 
corrects all errors of weight $\leq 
\lfloor \frac{v}{2s} \rfloor$ by majority-logic decoding  (i.e. Algorithm \ref{algo:bf} with $N=1$) when the 
 maximum 
column intersection is $s$. 
What makes this result interesting is  that for most MDPC codes the maximum column intersection is really small. We namely prove that for a natural 
random MDPC code model, the maximum column intersection  of the parity-check matrix defining the MDPC code is typically of
order $\BOm{\frac{\log n}{\log \log n}}$. Computing the maximum intersection number can obviously be performed in polynomial time and this allows us to give a randomized polynomial time algorithm for constructing MDPC codes of length $n$ and fixed rate $R \in (0,1)$ that corrects any pattern of 
$\BOm{\frac{\sqrt{n} \log \log n}{\log n}}$ errors with the majority-logic decoder.

Moreover, under a reasonable assumption on the first round of the bit-flipping decoder, the same MDPC codes correct $\BOm{\sqrt{n}}$ errors with two iterations of a bit-flipping decoder with decoding failure probability of order 
$e^{-\BOm{\frac{n \log \log n}{\log n}}}$.
It should be noted that under an additional assumption on the subsequent iterations of the bit-flipping decoder, it has been proved in \cite{ABBBBDGGGMMPSTZ17} that MDPC codes correct $\BOm{\sqrt{n}}$ 
 errors by performing an unbounded 
number of bit-flipping iterations  with probability of error $e^{-\BOm{n}}$.
We also provide some concrete numbers to show that it is possible to construct MDPC codes that avoid completely the  attack \cite{GJS16} and for which it is possible to provide a security proof in strong security models with a significant key size overhead when compared to the parameters proposed in \cite{MTSB13} if we want to stay in the no-error scenario, with
a reasonable overhead if we make the first assumption mentioned above, and  moderate overhead if we make  both assumptions mentioned above.

%% file: majoritylogic.tex
\section{Majority-logic decoding and its performance for MDPC codes}
We start this section by relating the error-correction capacity of an MDPC code to the maximum column intersection  of its defining  parity-check matrix,  then show that for typical MDPC codes the intersection number is small and that 
this allows to construct efficiently MDPC codes that correct all patterns of $\BO{\frac{\sqrt{n}\log \log n}{\log n}}$ errors.

\subsection{Error correction capacity of an MDPC code vs. maximum column intersection}

As explained in the introduction the maximum column intersection can be used to lower bound
the worst-case error-correction perfomance for the majority-logic decoder (Algorithm \ref{algo:bf} with 
$N=1$). The precise statement is given by the following proposition.
\begin{proposition}\label{prop:super_easy}
Consider a code with a parity check matrix for which every column has weight at least $v$ and whose maximum 
column intersection is $s$. Performing majority-logic decoding based on this matrix (i.e. Algorithm \ref{algo:bf} with $N=1$) corrects all errors of weight $\leq 
\lfloor \frac{v}{2s} \rfloor$.
\end{proposition}

\begin{proof}
We denote by $\Hm=(h_{ij})_{\substack{1 \leq i \leq r \\1 \leq j \leq n}}$ the parity-check matrix we use for performing majority-logic decoding and by $t$ the number of errors. We assume that $t \leq \lfloor \frac{v}{2s}\rfloor$.
For $i$ in $\{1,\dots,r\}$  denote by $E_{i}$ the subset of positions $\ell$ which are in error and in
 the support of the $i$-th parity check equation (i.e. $h_{i\ell}=1$). We number the parity-check equations of the code from $1$ to $r$.
We consider now what happens to $y_j$ in the algorithm.
There are two cases to consider. 

\noindent
{\bf Case 1:} $y_j$ is erroneous. 
We can upper-bound the number $s_j$ of satisfied parity-check equations involving this bit by the number of parity-check equations involving this bit whose 
support contains at least $2$ errors. 
We consider now the graph $G_j$ which is a bipartite graph associated to $j$ which is constructed as follows. Its set of vertices is the union of the set $A_j$ of positions different from $j$ which are in error and the set $B_j$ of parity-check equations that involve the position $j$ and whose support contains at least $2$ errors. There is an edge between a position $\ell$ in $A_j$ and parity-check equation $i$ in $B_j$ if and only if the parity-check equation involves $\ell$, 
that is $h_{i\ell}=1$.  Let $e_j$ be the number of edges of $G_j$ and let $n_j$ be the number of parity-check 
equations involving $j$. We observe now that
\begin{eqnarray}
s_j & \leq &\#\left\{ i: h_{ij} = 1, |E_i| \geq 2 \right\} \label{eq:un}\\
& \leq & e_j \nonumber \\
& \leq & s \# A_j \label{eq:deux} \\
& \leq & s (t-1) \nonumber \\
& \leq & s \left( \left\lfloor \frac{v}{2s} \right\rfloor -1\right) \nonumber\\
& < & v/2.\nonumber
\end{eqnarray}
\eqref{eq:un} is just the first observation whereas \eqref{eq:deux} follows from the fact the degree in $G_j$ of any vertex is at most $s$ 
by the assumption on the maximum intersection number of $\Hm$. Since $v/2 \leq n_j/2$ it follows that the majority-logic
decoder necessarily flips the bit  and therefore corrects the corresponding error.

{\bf Case 2:} there is no error in position $j$. We can upper-bound the number $u_j$ of unsatisfied positions in a similar way. This time
we consider the graph $G'_j$ whose vertex set is the union of $A'_j$ which is the set of positions which are in error and $B'_j$ the set 
of parity-check equations involving $j$ and whose support contains this time at least one error. We put an edge 
between a position $\ell$ in $A'_j$ and parity-check equation $i$ in $B'_j$ if and only if the parity-check equation involves $\ell$.
Let $e'_j$ be the number of edges of $G'_j$. Similarly to what we did we observe now that
\begin{eqnarray}
u_j & \leq &\#\left\{ i: h_{ij} = 1, |E_i| \geq 1 \right\} \nonumber\\
& \leq & e'_j \nonumber \\
& \leq & s \# A'_j \nonumber \\
& \leq & s t \nonumber \\
& \leq & s \left( \left\lfloor \frac{v}{2s} \right\rfloor \right) \\
& \leq & v/2 \nonumber \\
& \leq & n_j/2. \nonumber
\end{eqnarray}
In other words we will not flip this bit.
\end{proof}

\subsection{A random model for MDPC codes of type $(v,w)$}

There are several ways to build random MDPC codes of type $(v,w)$. The one which is used in cryptography \cite{BBC08,MTSB13,DGZ17,BGGMPST17,ABBBBDGGGMMPSTZ17,BBCPS17}
is to 
construct them as quasi-cyclic codes. Our proof technique can also be applied to this case, but since there are several different types of construction of this kind, so that we have to adapt our proof technique to each of those, we will consider a more general model here. It is based on Gallager's construction of LDPC codes \cite{G63}. We will construct an $r \times n$ random parity-check matrix of type $(v,w)$ by assuming 
that $n$ is a multiple of $w$ ($n = n'w$), $r$ is a multiple of $v$ ($r = r'v$) and that 
$rw=nv$ (this condition is necessary in order to obtain a matrix of type $(v,w)$).
Let $\Pm_{n,w}$ be a matrix of size $n' \times n$ constructed as follows
$$
\Pm_{n,w} = \Imat_{n'} \otimes \un_w = 
\begin{pmatrix}
\un_w & \zero_w & \hdots & \hdots & \zero_w \\
\zero_w & \un_w & \zero_w & \cdots & \zero_w \\
\vdots & \ddots & \ddots & \ddots & \vdots \\
\vdots & \hdots & \ddots & \ddots & \zero_w\\
\zero_w & \hdots & \hdots & \zero_w & \un_w
\end{pmatrix}.
$$
where $\Imat_{n'}$ denotes the identity matrix of size $n'$, $\un_{w}$ a row vector of length $w$ whose entries are all equal to $1$, that 
is $\un_w = \underbrace{(1 \dots 1)}_{w \text{ times}}$, $\zero_{w}$  
a row vector of length $w$ whose entries are all equal to $0$. We then choose $v$ permutations of length $n$ at random and they define
a parity-check matrix $\Hm(\pi_1,\dots,\pi_v)$ of size $r \times n$ of type $(v,w)$ as
$$
\Hm(\pi_1,\dots,\pi_v) = \begin{pmatrix}
\Pm_{n,w}^{\pi_1}\\
\Pm_{n,w}^{\pi_2}\\
\hdots \\
\Pm_{n,w}^{\pi_v}
\end{pmatrix},
$$
where $\Pm_{n,w}^{\pi_i}$ denotes the matrix $\Pm_{n,w}$ whose columns have been permuted with $\pi_i$.
We denote by ${\cD}_{r,n,v,w}$ the associated probability distribution of binary matrices of size $r \times n$
and type $(v,w)$ we obtain when the $\pi_i$'s are chosen uniformly at random.

\subsection{The maximum intersection number of matrices drawn according to ${\cD}_{r,n,v,w}$}
The maximum intersection number of matrices drawn according to ${\cD}_{r,n,v,w}$ turns out to be remarkably small
when $w$ and $v$ are of order $\sqrt{n}$, it is namely typically of order
$\BO{\frac{\log n}{\log \log n}}$. To prove this claim we first observe that
\begin{lemma}\label{lem:inp}
Consider a matrix $\Hm$ drawn at random according to the distribution ${\cD}_{r,n,v,w}$. Take two arbitrary columns $j$ and $j'$ of $\Hm$ and
let $n_{jj'}$ be the intersection number of $j$ and $j'$. We have for all $t \in \{0,\dots,v\}$
$$
\prob(n_{jj'}=t) = \binom{v}{t} \left( \frac{w-1}{n-1}\right)^t \left( 1 - \frac{w-1}{n-1} \right)^{v-t}.
$$
\end{lemma}

\begin{proof}
Recall that $\Hm=(h_{ij})_{\substack{1 \leq i \leq r \\1 \leq j \leq n}}$ is of the form
$$
\Hm(\pi_1,\dots,\pi_v) = \begin{pmatrix}
\Pm_{n,w}^{\pi_1}\\
\Pm_{n,w}^{\pi_2}\\
\hdots \\
\Pm_{n,w}^{\pi_v}
\end{pmatrix},
$$
for some permutations $\pi_1$, $\pi_2, \dots,\pi_v$ chosen uniformly at random in  $S_n$.
A row $i$ of $\Hm$ is called a coincidence if and only if $h_{ij}=h_{ij'}$. There is obviously one coincidence
at most in each of the blocks $\Pm_{n,w}^{\pi_\ell}$. We claim now that the probability of a coincidence in each of these blocks is 
$\frac{w-1}{n-1}$. To verify this consider the row $i$ of block $\Pm_{n,w}^{\pi_\ell}$ which is such that $h_{ij}=1$. The probability that there is
a coincidence  for this block is the probability that $h_{ij'}=1$ which amounts to the fact that $\pi_{\ell}(j')$ takes its values
in a subset of $w-1$ values among $n-1$ possible values. All of these $n-1$ are equiprobable.
This shows the claim. Since the coincidences that occur in the blocks are all independent (since the
$\pi_i$'s are independent) we obtain the aforementioned formula.
\end{proof}

We use this to prove the following result
\begin{proposition}
Let $\alpha$ and $\beta$ be two constants such that $0<\alpha < \beta$. Assume we draw a parity-check
matrix $\Hm$  at random 
according
to the distribution ${\cD}_{r,n,v,w}$ where we assume that both $v$ and $w$ satisfy
$\alpha \sqrt{n} \leq v < w \leq \beta \sqrt{n}$. Then for any $\varepsilon >0$ 
the maximum intersection number of $\Hm$ is smaller than
$(2+\varepsilon) \frac{\ln n}{\ln \log n}$ with probability $1-o(1)$ as $n$ tends
to infinity.
\end{proposition}

\begin{proof}
Let is number the columns of $\Hm$ from $1$ to $n$.
For $i$ and $j$ in $\{1,\dots,n\}$  two different columns of $\Hm$
we denote by $E_{i,j,t}$ the event that the intersection number of $i$ and $j$ is $\geq t$.
Let $E_t$ be the probability that the maximum intersection number is larger than or equal to $t$.
By the union bound, and then Lemma \ref{lem:inp} we obtain
\begin{eqnarray}
\prob(E_t) & = & \prob\left( \bigcup_{1 \leq i < j \leq n} E_{i,j,t} \right) \nonumber\\
& \leq & \sum_{1 \leq i < j \leq n} \prob(E_{i,j,t}) \nonumber \\
& \leq & n^2 \sum_{s=t}^v  \binom{v}{s} \left( \frac{w-1}{n-1}\right)^s \left( 1 - \frac{w-1}{n-1} \right)^{v-s} \nonumber \\
 \nonumber
\end{eqnarray}
From this we deduce
$$
\prob(E_t) 
 \leq  n^2 \sum_{s=t}^v  \frac{v^v}{s^s (v-s)^{v-s}} \left( \frac{w-1}{n-1}\right)^s \left( 1 - \frac{w-1}{n-1} \right)^{v-s} 
 $$
 where we use the well known upper-bound
$\binom{v}{s} \leq e^{v h(s/v)}=\frac{v^v}{s^s (v-s)^{v-s}}$. This allows to write
 \begin{eqnarray}
\prob(E_t) & \leq & n^2 \sum_{s=t}^v \frac{v^v}{s^s (v-s)^{v-s}} \left( \frac{w}{n} \right)^s \label{eq:trois}\\
& \leq & n^2 \sum_{s=t}^v \frac{v^{v-s}}{(v-s)^{v-s}} \left(\frac{v \cdot w}{s \cdot n}\right)^s \nonumber\\
& \leq & n^2 \sum_{s=t}^v \left( 1 + \frac{s}{v-s} \right)^{v-s} \left(\frac{v \cdot w}{s \cdot n}\right)^s\nonumber\\
& \leq & n^2 \sum_{s=t}^v \left(\frac{e \cdot v \cdot w}{s\cdot n}\right)^s \nonumber\\
& \leq & n^2 \sum_{s=t}^v \left(\frac{e \beta^2}{s}\right)^s \nonumber
\end{eqnarray}
Choose now $t \geq (2+\varepsilon) \frac{\ln n}{\ln \ln n}$ for some $\varepsilon >0$. 
When $n$ is large enough, we have that $\frac{e \beta^2}{t} < 1$. In such a case we can write
\begin{eqnarray*}
\prob(E_t) & \leq & n^2 \sum_{s=t}^v \left(\frac{e \beta^2}{t}\right)^s\\
& \leq & \frac{n^2 \left(\frac{e \beta^2}{t}\right)^t}{1 - \frac{e \beta^2}{t}}\\
& \leq &  n^2 \left( \frac{K}{t} \right)^t 
\end{eqnarray*}
for some constant $K>0$. This implies that 
\begin{eqnarray*}
\prob(E_t) & \leq & n^2 e^{(2+\varepsilon) \frac{\ln n}{\ln \ln n} \ln \left(\frac{K \ln \ln n}{\gamma \ln n} \right)}\\
& \leq &  e^{\varepsilon \ln n + \frac{\gamma \ln n \ln \left(\frac{K \ln \ln n}{\gamma}\right)}{\ln \ln n}}\\
& = & o(1)
\end{eqnarray*}
as $n$ tends to infinity.
\end{proof}

This together with Proposition \ref{prop:super_easy} implies directly the following corollary
\begin{corollary}\label{cor:construction}
There exists a randomized algorithm working in expected polynomial time outputting for any designed rate $R \in (0,1)$ an MDPC code of rate $\geq R$ of an arbitrarily large length $n$ and parity-check equations of weight $\Theta(\sqrt{n})$ that corrects all patterns of errors of size less than $\gamma \frac{\sqrt{n} \ln \ln n}{\ln n}$ 
for $n$ large enough, where $\gamma >0$ is some absolute constant.
\end{corollary}

\begin{proof}
The randomized algorithm is very simple. We choose $n$ to be a square $n=w^2$ for some integer $w$ and let $v \eqdef \left\lfloor (1-R)w \right\rfloor$ and 
$r \eqdef \frac{nv}{w}$. 
Then we draw a parity-check
matrix $\Hm$  at random 
according
to the distribution ${\cD}_{r,n,v,w}$.
The corresponding code has clearly rate $\geq R$.
We compute the maximum column intersection of $\Hm$.  This can be done in time $\BO{wn^2}$. If this column intersection is greater than
$(2+\varepsilon) \frac{\ln n}{\ln \ln n}$ we output the corresponding MDPC code, if not we draw at random $\Hm$ again until finding a suitable
matrix $\Hm$. By Proposition \ref{prop:super_easy} we know that such a code can correct all patterns of at most 
$\left\lfloor \frac{\alpha \sqrt{n}\ln \ln n}{(4+2\varepsilon) \ln n} \right\rfloor$ errors. This implies the corollary.
 \end{proof}

%% file: analysis.tex
\section{Analysis of two iterations of bit-flipping}

We derived in the previous section a condition ensuring that one round of bit-flipping corrects all the errors. 
We will now estimate the probability that performing
one round of bit-flipping corrects enough errors so that another round of bit-flipping will correct all remaining errors. To analyze the first round of decoding we will model intermediate quantities of the bit-flipping algorithm by  binomial distributions.
More precisely, consider an MDPC code of type $(v,w)$ and length $n$.
 The noise model is the following: an error of weight $t$ was chosen uniformly at random and added to the
codeword of the MDPC code. For $i \in \{1,\dots,n-t\}$, let $E^{0}_i$ be the Bernouilli random variable which is equal to $1$ iff the $i$-th position that was not in error initially is in error after the first round of iterative decoding. We also denote by $U^0_i$ the counter $u_j$ associated to the 
$i$-th position which was not in error. $U_i^0$ is the sum $\sum_{j=1}^vV^0_{ij}$ of $v$ Bernouilli random variables $V^0_{ij}$ associated to the $v$ parity-check equations involving this bit. A Bernoulli-random variable $V^0_{ij}$ is equal to $1$ if and only the corresponding parity-check is equal to $1$.
Note that by definition of the bit-flipping decoder
$$
E_i^0 = \un_{\{U_i^0 > v/2\}}
$$
Similarly, for $i \in \{1,\dots,t\}$ we denote by 
$E^1_i$ the Bernoulli random variable that is equal to $1$ iff the $i$-th bit that was in error initially stays in error after 
the first round of Algorithm \ref{algo:bf}. We also define the $U^1_i$'s and the $V^1_{ij}$'s similarly. In this case 
$$
E_i^1 = \un_{\{U_i^1 \leq v/2\}}.
$$
Let us bring in for $b \in \{0,1\}$:
\begin{equation}
\label{e:defpb}
p_b \eqdef \prob(V^{b}_{ij}=1).
\end{equation}
It is clear that these probabilities do not depend on $i$ and $j$ and that this definition is consistent. It is (essentially) proved in \cite{ABBBBDGGGMMPSTZ17} that 

\begin{restatable}{lemma}{lempb}
\label{lem:pb}
Assume that $w = \BO{\sqrt{n}}$ and $t = \BO{\sqrt{n}}$. Then
\begin{equation}
p_b = \frac{1}{2} - (-1)^b \varepsilon \left(\frac{1}{2}+\BO{\frac{1}{\sqrt{n}}}\right),
\end{equation}
where $\varepsilon \eqdef e^{-\frac{2wt}{n}}$.
\end{restatable}
We will recall a proof of this statement in the appendix.
We will now make the following assumption that simplifies the analysis
\begin{assumption}\label{ass:basic}
When we use Algorithm \ref{algo:bf} on an MDPC code of type $(v,w)$, we assume that 
\begin{itemize}
\item
for all $i \in \{1,\dots,n-t\}$ the counters $U_i^0$ of Algorithm \ref{algo:bf} are distributed like sums of $v$ independent Bernoulli random variables of parameter $p_0$ at the first iteration and the $E^0_i$'s are independent;
\item for all $i \in \{1,\dots,t\}$ the counters $U_i^1$ of Algorithm \ref{algo:bf} are distributed like sums of $v$ independent Bernoulli random variables of parameter $p_1$ at the first iteration and the $E^1_i$'s are independent.
\end{itemize}
\end{assumption}
The experiments performed in \cite{C17} corroborate this assumption for the first iteration of the bit-flipping decoder.
To analyze the behavior of Algorithm \ref{algo:bf} we will use the following lemma which is just a slight generalization of Lemma 6 in \cite{ABBBBDGGGMMPSTZ17}
\begin{restatable}{lemma}{lemqb}
\label{lem:qb}
Under Assumption \ref{ass:basic} used for an MDPC code of type $(v,w)$ and when the error is chosen uniformly at random among the errors of 
weight $t$, we have for all
 $(b,i) \in \{0\}\times\{1,\ldots,n-t\} \cup \{1\}\times \{1,\ldots,t\}$,
$$
\prob(E^b_i=1) = \BO{\dfrac{(1-\varepsilon^2)^{v/2}}{\sqrt{v} \varepsilon}},
$$
where $\varepsilon \eqdef e^{-\frac{2wt}{n}}$.
\end{restatable}

For the ease of reading the proof of this lemma is also recalled in the appendix.
Under Assumption \ref{ass:basic}, $\prob(E^b_i=1)$ does not depend on $i$, we will  denote it by
$$
q_b \eqdef \prob(E^b_i=1).
$$
We let
\begin{eqnarray*}
S_0 & \eqdef & E^0_1 + \dots + E^0_{n-t}\\
S_1 & \eqdef & E^1_1 + \dots + E^1_{t}
\end{eqnarray*}
$S_0$ is the number of errors that were introduced after one round of iterative decoding coming from flipping the $n-t$ bits that were initially correct.
Similarly $S_1$ is the number of errors that are left after one round of iterative decoding coming from not flipping the $t$ bits that were initially incorrect.
Let $S \eqdef S_0+S_1$, which represents the total number of errors that are left after the first round of iterative decoding.
We quantify the probability that this quantity does not decay enough by the following theorem which holds under Assumption \ref{ass:basic}.
\begin{restatable}{theorem}{thAsymptotic}
\label{th:asymptotic}
Under Assumption \ref{ass:basic}, we have for an MDPC code of type $(v,w)$ where $v=\Theta(\sqrt{n})$ and
$w=\Theta(\sqrt{n})$:
$$
\prob(S \geq t') \leq \dfrac{1}{\sqrt{t'}} e^{\frac{t' v}{4} \ln\left(1-\varepsilon^2 \right) + \frac{t'}{8} \ln\left(n \right) + \BO{t'\ln(t'/t)}},
$$
where $\varepsilon \eqdef e^{-\frac{2wt}{n}}$.
\end{restatable}

From this theorem we deduce that
\begin{corollary}
\label{cor:asymptotic}
Provided that Assumption \ref{ass:basic} holds, we can construct in expected polynomial time 
for any designed rate $R \in (0,1)$ an MDPC code of rate $\geq R$ of an arbitrarily large length $n$ and 
parity-check equations of weight $\Theta(\sqrt{n})$  where the probability of error $P_e$ after two iterations of bit-flipping is upper-bounded by $e^{- \BOm{n \frac{\ln \ln n}{\ln n}}} $
when there are $t= \Theta(\sqrt{n})$ errors.
\end{corollary}

\begin{proof}
We use the construction given in the proof of Corollary \ref{cor:construction} to construct an MDPC code of type $(v,w)$ of length $n=w^2$ and with 
$v \eqdef \left\lfloor (1-R)w \right\rfloor$ that allows to correct all patterns of errors of size less than $\gamma \frac{\sqrt{n} \ln \ln n}{\ln n}$ 
for $n$ large enough, where $\gamma >0$ is some absolute constant with just one round of the bit-flipping decoder  of Algorithm \ref{algo:bf}. 
Then we use Theorem \ref{th:asymptotic} to show that with probability  upper-bounded by $e^{- \BOm{n \frac{\ln \ln n}{\ln n}}} $ there remains at most 
$\gamma \frac{\sqrt{n} \ln \ln n}{\ln n}$ errors after one round of Algorithm \ref{algo:bf}. This proves the corollary.
\end{proof}

In \cite{ABBBBDGGGMMPSTZ17} there is an additional assumption which is made which is that the probability of error is dominated by the 
probability that the first round of decoding is not able to decrease the number by some mutiplicative factor $\alpha$. With the notation of 
this section, this assumption can be described as follows.
\begin{assumption}
\label{ass:additional}
There exists some constant $\alpha>0$ such that
the probability of error $P_{\text{err}}$ for an unbounded number of iterations of Algorithm \ref{algo:bf} is upper-bounded by 
$\prob(S \geq \alpha t)$ where $S$ is the number of errors that are left after the first round of Algorithm \ref{algo:bf}  and $t$ is the initial number of errors. 
\end{assumption}
This assumption also agrees with the experiments performed in \cite{C17}.
With this additional assumption (Assumption \ref{ass:basic} is actually also made) it is proven in \cite{ABBBBDGGGMMPSTZ17} that the probability of errors decays exponentially when 
$t=\Theta(\sqrt{n})$. This is actually obtained by a slightly less general version of Theorem \ref{th:asymptotic} (see \cite[Theorem 1]{ABBBBDGGGMMPSTZ17}).

%% file: numerical.tex
\section{Numerical results}
\label{sec:numerical}

In this section we provide numerical results showing how much we have to increase the parameters proposed in 
\cite{MTSB13} in order to obtain a probability of error which is below $2^{-\lambda}$ where $\lambda$ is the security 
parameter (i.e. $2^\lambda$ should be the complexity of the best attacks on the scheme). The upper-bound on the probability that the maximum column intersection is larger than some bound coming from using Lemma \ref{lem:inp} together with an obvious union-bound is a little bit loose, and we performed numerical tests in order to estimate the maximum
column intersection. To speed-up the calculation and at the same time to be closer to the cryptographic
applications we considered the particular code structure used in \cite{MTSB13,BGGMPST17,ABBBBDGGGMMPSTZ17}, namely a quasi-cyclic code whose parity-chack matrix $\Hm$ is formed by two circulant blocks $\Hm_0$ and $\Hm_1$, i.e. $\Hm = \begin{pmatrix} \Hm_0 & \Hm_1 \end{pmatrix}$.
The weight of the rows of $\Hm_0$ and $\Hm_1$ was chosen to be $w/2$ (with $w$ even) so that we have
a code of type $(w/2,w)$.
The maximum column intersection $s$ given in Table \ref{tab:numerical} corresponds to the smallest number 
$s_0$ such that more than $20\%$ of the parity-check matrices had a maximum column intersection $\leq s_0$.
We considered several scenarios:\\
- Scenario I, Algorithm \ref{algo:bf} with $N=1$ and a zero-error decoding failure $P_{\text{err}}$ 
probability when there are $t$ errors;  \\
- Scenario II: Algorithm \ref{algo:bf} with $N=2$ and a non-zero decoding failure probability 
upper-bounded by making Assumption \ref{ass:basic}  when there are $t$ errors, \\
- Scenario III, Algorithm \ref{algo:bf} with $N>2$  and a non-zero decoding failure probability  when there are $t$ errors
which is upper-bounded by making Assumption \ref{ass:basic} and Assumption \ref{ass:additional} for $\alpha=0.5$;\\
- Scenario IV, Algorithm \ref{algo:bf} with $N>2$  and a non-zero decoding failure probability  when there are $t$ errors
which is upper-bounded by making Assumption \ref{ass:basic} and Assumption \ref{ass:additional} for $\alpha=0.75$.
This should be compared with the original parameters proposed for a security level $\lambda=80$ in \cite{MTSB13} that were broken in \cite{GJS16}, namely $n=9602$, $w=90$, $t=84$. We have chosen $t=84$ in all cases.
\begin{table}[h!]
\caption{Numerical results}
\begin{center}
\begin{tabular}{|c|c|c|c|c|c}
\hline
Scenario & $n$ & $w$ &$s$ & $\Perr$ \\ \hline
I & 4100014 & 4034 & 12 & 0 \\ \hline
II & 35078 & 110 & 9 & $2^{-80}$ \\ \hline
III & 20854 & 90 & 3 & $2^{-80}$ \\ \hline
IV & 18982& 90 & 3 & $2^{-80}$\\ \hline
\end{tabular}
\end{center}
\label{tab:numerical}
\end{table}

%% file: conclusion.tex
\section{Concluding remarks}
This study shows that it is possible to devise MDPC codes with  zero or very small error probability, and in the last case
it comes at an affordable cost for cryptographic applications and this by making assumptions that are corroborated by experimental evidence \cite{C17}. There are obviously several ways to improve our results.
The first would be to use slightly more sophisticated decoding techniques and/or more sophisticated analyses when we want a zero-error probability. The maximum column intersection gives a lower bound on the expansion of the Tanner graph and this can be used to study the bit-flipping algorithm considered in \cite{SS96}. This would not improve the lower-bound on the error-correction capacity however, but suggests that refined considerations and decoding algorithms should be able to improve the error-correction capacity in the worst case. Moreover, in order to simplify the analysis and the discussion we have considered a very simple decoder. The probability of error can already be lowered rather significantly by choosing in a slightly better way the threshold in Step \ref{step:threshold} in Algorithm \ref{algo:bf} and it is clear that more sophisticated decoding techniques will be able to lower the probability of error significantly (see \cite{C17} for instance). This suggests that it should be possible to improve rather significantly the  parameters proposed in Table \ref{tab:numerical}.

%% file: appendix.tex
%
%
\section{The Kullback-Leibler divergence}
 The proofs of the results proved in the appendix use the Kullback-Leibler divergence (see see for instance \cite{CT91}) and some of its properties what we now recall.
 
  \begin{definition}{\bf{Kullback-Leibler divergence}\\}
  Consider two discrete probability distributions $\pv$ and $\qv$ defined over a same discrete space
  $\mathcal{X}$. The Kullback-Leibler divergence between $\pv$ and $\qv$ is defined by 
  $$
  \Kull{\pv}{\qv} = \sum_{x \in \mathcal{X}} p(x) \ln \frac{p(x)}{q(x)}.
  $$
  We overload this notation by defining for two Bernoulli  distributions $\mathcal{B}(p)$ and 
  $\mathcal{B}(q)$ of respective parameters $p$ and $q$
  $$
  \Kull{p}{q} \eqdef \Kull{\mathcal{B}(p)}{\mathcal{B}(q)} = p\ln\left(\frac{p}{q}\right)+(1-p)\ln\left(\frac{1-p}{1-q}\right).
  $$
    We use the convention (based on continuity arguments) that $0\ln\frac{0}{p}=0$ and $p\ln\frac{p}{0}=\infty$.
  \end{definition}

We will need the following approximations/results of the Kullback-Leibler divergence

\begin{lemma}\label{lem:DK}
For any $\delta \in (-1/2,1/2)$ we have
\begin{equation}
\label{eq:KL12}
\Kull{\frac{1}{2}}{\frac{1}{2}+\delta} = - \frac{1}{2} \ln(1-4\delta^2).
\end{equation}
For constant $\alpha \in (0,1)$ and $\delta$ going to $0$ by staying positive, we have
\begin{equation}
\label{eq:alphaepsilon}
\Kull{\alpha}{\delta} = - h(\alpha) - \alpha \ln \delta +O(\delta).
\end{equation}
For $0<y<x$ and $x$ going to $0$ we have
\begin{equation}
\label{eq:KLxy}
\Kull{x}{y} = x \ln \frac{x}{y} + x - y + \BO{x^2}.
\end{equation}
\end{lemma}

\begin{proof}
Let us first prove \eqref{eq:KL12}.
\begin{eqnarray*}
\Kull{\frac{1}{2}}{\frac{1}{2}+\delta} &=& \frac{1}{2} \ln \frac{1/2}{1/2+\delta} +
\frac{1}{2} \ln \frac{1/2}{1/2-\delta} \\
\prob& = & - \frac{1}{2} \ln (1+2\delta) - 
\frac{1}{2} \ln (1-2\delta) \\
& = & 
- \frac{1}{2} \ln(1-4\delta^2).
 \end{eqnarray*}
 To prove \eqref{eq:alphaepsilon} we observe that
 \begin{eqnarray*}
 \Kull{\alpha}{\delta} &= &
  \alpha\ln\left(\frac{\alpha}{\delta}\right)+(1-\alpha)\ln\left(\frac{1-\alpha}{1-\delta}\right) \\
  & = & -h(\alpha) - \alpha \ln \delta -(1-\alpha) \ln(1- \delta) \\
& = &  - h(\alpha) - \alpha \ln \delta +O(\delta).
 \end{eqnarray*}
 For the last estimate we proceed as follows
 \begin{eqnarray*}
 \Kull{x}{y} & = & x \ln \frac{x}{y} + (1-x) \ln \frac{1-x}{1-y} \\
 & = & x \ln \frac{x}{y} - (1-x) \left( -x + y + \BO{x^2}\right) \\
 & = & x \ln \frac{x}{y}  + x - y + \BO{x^2}.
 \end{eqnarray*}
\end{proof}

The Kullback-Leibler appears in the computation of large deviation exponents. In our case, we will use
the following estimate which is well known and which can be found for instance in 
\cite{BGT11}
\begin{lemma}\label{lem:sum}
  Let $p$ be a real number in $(0,1)$ and $X_1, \dots X_n$ be $n$ independent 
  Bernoulli random variables of parameter $p$. Then, as $n$ tends to infinity:
    \begin{eqnarray}
      \label{eq:asymptotic1}
      \prob(X_1 + \dots X_n  \geq \tau n) &= &
      \dfrac{(1-p)\sqrt{\tau}}{(\tau-p)\sqrt{2\pi n (1-\tau)}}e^{-n \Kull{\tau}{p}} (1+o(1))
      \;\text{for $p <\tau < 1$,}\\
       \label{eq:asymptotic2}
      \prob(X_1 + \dots X_n  \leq \tau n) 
      &= &
      \dfrac{p \sqrt{1-\tau}}{(p-\tau)\sqrt{2\pi n \tau}}e^{-n \Kull{\tau}{p}}(1+o(1))
       \;\text{for $0< \tau < p$.}
    \end{eqnarray}
\end{lemma}

\section{Proof of Lemma \ref{lem:pb}}

Recall first this lemma.
\lempb*

Before giving the proof of this lemma, observe $\prob(V_{ij}^b=1)$ can be viewed as the probability that the $j$-th parity 
check equation involving a bit $i$ gives an incorrect information about bit $i$.
This is obtained through the following lemma.
\begin{lemma}
Consider a word $\hv \in \F{2}^n$ of weight $w$ and an error $\ev \in \F{2}^n$ of weight $t$ chosen uniformly at random. Assume that 
both $w$ and $t$ are of order $\sqrt{n}$: $w = \BO{\sqrt{n}}$ and $t = \BO{\sqrt{n}}$. We have
$$
\prob_{\ev} (\langle \hv, \ev \rangle=1) = \frac{1}{2} - \frac{1}{2} e^{-\frac{2wt}{n}}\left(1+\BO{\frac{1}{\sqrt{n}}}\right).
$$
\label{lem:paritycheck}
\end{lemma}

\begin{remark}
Note that this probability is in this case of the same order as the probability taken over errors $\ev$ whose
coordinates are drawn independently from a Bernoulli distribution of parameter $t/n$. In such a case, 
from the piling-up lemma \cite{M93} we have
\begin{eqnarray*}
\prob_{\ev} (\langle \hv, \ev \rangle=1) &= & \frac{1 - \left( 1- \frac{2t}{n}\right)^w}{2} \\
& = & \frac{1}{2} - \frac{1}{2} e^{w \ln(1-2t/n)}\\
& = & \frac{1}{2} - \frac{1}{2} e^{-\frac{2wt}{n}}\left(1+\BO{\frac{1}{\sqrt{n}}}\right).
\end{eqnarray*}
\end{remark}

The proof of this lemma will be done in the following subsection. Lemma \ref{lem:pb} is a corollary of this lemma since we have
\begin{equation}
p_b  =  \prob(\langle \hv, \ev \rangle=1|e_1=b).
\end{equation}

\subsection{Proof of Lemma \ref{lem:paritycheck}}

The proof involves properties of the Krawtchouk polynomials. We  recall that the (binary) Krawtchouk polynomial of degree $i$ and order $n$ (which is an integer), $P_i^n(X)$ is  defined for $i  \in \{0,\cdots,n\}$ by:
\begin{equation}\label{eq:Krawtchouk} P_{i}^{n}(X) \eqdef \frac{(-1)^{i}}{2^{i}}  \sum_{j=0}^{i} (-1)^{j} \binom{X}{j}  \binom{n-X}{i-j} \quad \mbox{where } \binom{X}{j} \eqdef \frac{1}{j!} X  (X-1) \cdots (X-j+1).
\end{equation}
Notice that it follows on the spot from the definition
of a Krawtchouk polynomial that
\begin{equation}
\label{eq:0}
P_k^n(0) =\frac{(-1)^k \binom{n}{k}}{2^k}.
\end{equation}
	
Let us define the bias $\delta$ by 
$$
\delta \eqdef 1 - 2 \prob_{\ev} (\langle \hv, \ev \rangle=1).
$$	
In other words $\prob_{\ev} (\langle \hv, \ev \rangle=1)=\frac{1}{2}(1-\delta)$.
These Krawtchouk polynomials are readily related to $\delta$.
We first observe that
\begin{eqnarray*}
\prob_{\ev} (\langle \hv, \ev \rangle=1) = \frac{\sum_{\substack{j=1 \\ j \text{ odd}}}^{w} \binom{t}{j}  \binom{n-t}{w-j}}{\binom{n}{w}}.
\end{eqnarray*}
Moreover by observing that $\sum_{j=0}^{w} \binom{t}{j}  \binom{n-t}{w-j} = \binom{n}{w}$ we can recast the following 
evaluation of a Krawtchouk polynomial as
\begin{eqnarray}
 \frac{(-2)^{w}}{\binom{n}{w}}  P_{w}^{n}(t) &= & \frac{\sum_{j=0}^{w} (-1)^{j} \binom{t}{j}  \binom{n-t}{w-j}}{ \binom{n}{w}}
\nonumber\\
& = & \frac{\sum_{\substack{j=0 \\ j \text{ even}}}^{w}  \binom{t}{j}  \binom{n-t}{w-j}
- \sum_{\substack{j=1 \\ j \text{ odd}}}^{w}  \binom{t}{j}  \binom{n-t}{w-j} }{\binom{n}{w}} \nonumber\\
& = & \frac{\binom{n}{w} -  2 \sum_{ \substack{j=1 \\ j \text{ odd}}}^{w}  \binom{t}{j}  \binom{n-t}{w-j}
  }{\binom{n}{w}}\nonumber\\
& = & 1 - 2 \prob_{\ev} (\langle \hv, \ev \rangle=1) \nonumber\\
& = & \delta.  \label{eq:delta}
\end{eqnarray}
To simplify notation we will drop the superscript $n$ in the Krawtchouk polynomial notation. It will be chosen 
as  the length of the MDPC code when will use it in our case. 
An important lemma that we will need is the following one.
\begin{lemma}\label{lem:Krawtchouk}
For all $x$ in $\{1,\dots,t\}$, we have
$$
\frac{P_w(x)}{P_w(x-1)} = \left( 1+ \BO{\frac{1}{n}}\right) \frac{n-2w+\sqrt{(n-2w)^2-4w(n-w)}}{2(n-w)}.
$$
\end{lemma}
\begin{proof}
This follows essentially from arguments taken in the proof of \cite{MS86}[Lemma 36, \S 7, Ch. 17]. The result we use appears however more explicitly
in \cite{KL95}[Sec. IV] where it is proved that if $x$ is in an interval of the form
$\left[0,(1-\alpha)\left(n/2 -\sqrt{w(n-w)}\right)\right]$ for some constant $\alpha \in [0,1)$ independent of $x$, $n$ and $w$, then
$$
\frac{P_w(x+1)}{P_w(x)} = \left( 1+ \BO{\frac{1}{n}}\right) \frac{n-2w+\sqrt{(n-2w)^2-4w(n-w)}}{2(n-w)}.
$$
For our choice of $t$ this condition is met for $x$ and the lemma follows immediately.
\end{proof}

We are ready now to prove Lemma \ref{lem:paritycheck}.
\begin{proof}[Proof of Lemma \ref{lem:paritycheck}]
We start the proof by using \eqref{eq:delta} which says that
$$
\delta = \frac{(-2)^{w}}{\binom{n}{w}} P_{w}^{n}(t).
$$
We then observe that
\begin{eqnarray*}
\frac{(-2)^{w}}{\binom{n}{w}} P_{w}^{n}(t) & = & \frac{(-2)^{w}}{\binom{n}{w}} \frac{P_{w}^{n}(t)}{P_{w}^{n}(t-1)}\frac{P_{w}^{n}(t-1)}{P_{w}^{n}(t-2)}\dots\frac{P_{w}^{n}(1)}{P_{w}^{n}(0)} P_{w}^{n}(0)\\
& = &\frac{(-2)^{w}}{\binom{n}{w}} \left(\left( 1+ \BO{\frac{1}{n}}\right) \frac{n-2w+\sqrt{(n-2w)^2-4w(n-w)}}{2(n-w)}\right)^t P_{w}^{n}(0) \text{ (by Lemma \ref{lem:Krawtchouk})}\\
& = &  \left( 1+ \BO{\frac{1}{n}}\right)^t \left( \frac{n-2w+\sqrt{(n-2w)^2-4w(n-w)}}{2(n-w)}\right)^t \text{(by \eqref{eq:0})}\\
& = & e^{t\ln \left( \frac{1-2\omega+\sqrt{(1-2\omega)^2-4\omega(1-\omega)}}{2(1-\omega)}\right)} \left( 1+ \BO{\frac{t}{n}}\right) \text{where $\omega \eqdef \frac{w}{n}$}\\
& = & e^{t \ln \left( \frac{1-2\omega+1-4\omega+\BO{\omega^2}}{2(1-\omega)}\right)}\left( 1+ \BO{\frac{t}{n}}\right)\\
& = & e^{t\ln \left( \frac{1-3\omega+\BO{\omega^2}}{1-\omega}\right)}\left( 1+ \BO{\frac{t}{n}}\right)\\
&=  & e^{-2t\omega+\BO{\frac{tw^2}{n^2}}}\left( 1+ \BO{\frac{t}{n}}\right)\\
& = & e^{-\frac{2wt}{n}}\left(1+\BO{\frac{1}{\sqrt{n}}}\right),
\end{eqnarray*}
where we used at the last equation that $t = \BO{\sqrt{n}}$ and $w=\BO{\sqrt{n}}$.
\end{proof}
\section{Proof of Lemma \ref{lem:qb}}

Let us first recall this lemma.
\lemqb*

\begin{proof}
For
 $(b,i) \in \{0\}\times\{1,\ldots,n-t\} \cup \{1\}\times \{1,\ldots,t\}$, let $X_i$ be independent Bernoulli random variables of parameter $p_b$.
 From Assumption \ref{ass:basic} we have
\begin{eqnarray*}
\prob(E^0_i=1) =q_0 & = & \prob(\sum_{i=1}^{v} X^0_i > v/2)\\
\prob(E^1_i=1)= q_1 & = &  \prob(\sum_{i=1}^{v} X^1_i \leq v/2).
\end{eqnarray*}
By using Lemma \ref{lem:sum} we obtain for $q_0$
   \begin{eqnarray}
 q_0 & \leq & 
      \dfrac{(1-p_0)\sqrt{\frac{1}{2}}}{(\frac{1}{2}-p_0)\sqrt{2\pi v (1-\frac{1}{2})}}e^{- v \Kull{\frac{1}{2}}{p_0}} 
      \nonumber \\
      & \leq & 
      \dfrac{(1-p_0)\sqrt{2}}{\sqrt{\pi v} \varepsilon \left(1+ \BO{1/\sqrt{n}} \right)}e^{-v \Kull{\frac{1}{2}}{\frac{1}{2} - \frac{1}{2} \varepsilon \left(1+ \BO{1/\sqrt{n}} \right)}} \\
      & \leq &  \dfrac{(1-p_0)\sqrt{2}}{\sqrt{\pi v} \varepsilon \left(1+ \BO{1/\sqrt{n}} \right)}e^{\frac{v \left(\ln(1- \varepsilon^2) + \BO{\frac{1}{\sqrt{n}}}\right)}{2}} \\
      & \leq &  \BO{\dfrac{(1-\varepsilon^2)^{v/2}}{\sqrt{v} \varepsilon}}
      \end{eqnarray}
Whereas for $q_1$ we also obtain      
      \begin{eqnarray}
      q_1 & \leq & 
      \dfrac{p_1 \sqrt{\frac{1}{2}}}{(p_1-\frac{1}{2})\sqrt{2\pi v \frac{1}{2}}}e^{-v \Kull{\frac{1}{2}}{p_1}}\\
      & \leq &
       \BO{\dfrac{(1-\varepsilon^2)^{v/2}}{\sqrt{v} \varepsilon}}
    \end{eqnarray}
\end{proof}

\section{Proof of Theorem \ref{th:asymptotic}}

We are ready now to prove Theorem \ref{th:asymptotic}.
We first recall it.

\thAsymptotic*

\begin{proof}
\begin{eqnarray*}
\prob(S \geq t') & \leq & \prob(S_0 \geq t'/2 \cup S_1 \geq t'/2 )\\
& \leq & \prob(S_0 \geq t'/2) +  \prob(S_1 \geq t'/2)
\end{eqnarray*}
By Assumption \ref{ass:basic},  $S_0$ is the sum of $n-t$ Bernoulli variables of parameter $q_0$. By applying Lemma
\ref{lem:sum} we obtain
\begin{eqnarray}
 \prob(S_0 \geq t'/2) & \leq & 
 \dfrac{(1-q_0)\sqrt{\frac{t'}{2(n-t)}}}{(\frac{t'}{2(n-t)}-q_0)\sqrt{2\pi (n-t) (1-\frac{t'}{2(n-t)})}}e^{-(n-t) \Kull{\frac{t'}{2(n-t)}}{q_0}} \nonumber \\
 & \leq & 
 \BO{\dfrac{1}{\sqrt{ t'}} e^{-(n-t) \Kull{\frac{t'}{2(n-t)}}{q_0}}} \label{eq:S0}
\end{eqnarray}

We observe now that 
\begin{equation}
\Kull{\frac{t'}{2(n-t)}}{q_0} \geq  \Kull{\frac{t'}{2(n-t)}}{\BO{\frac{(1-\varepsilon^2)^{v/2}}{\sqrt{v} \varepsilon}}}
\end{equation}
where we used the upper-bound on $q_0$ coming from Lemma \ref{lem:qb} and the fact that $\Kull{x}{y} \geq \Kull{x}{y'}$ for 
$0<y<y'<x<1$.
By using this and Lemma \ref{lem:DK}, we deduce
\begin{eqnarray*}
\Kull{\frac{t'}{2(n-t)}}{q_0} &\geq & \frac{t'}{2(n-t)}\ln \left(\frac{t'}{2(n-t)} \right) - \frac{t'}{2(n-t)}
\ln \left( \BO{\frac{(1-\varepsilon^2)^{v/2}}{\varepsilon \sqrt{v}}} \right) +\BO{\frac{t'}{2(n-t)}}\\
& \geq & \frac{t'}{2(n-t)} \ln\left(\frac{t'\sqrt{v}}{n}\right)- \frac{t' v}{4(n-t)}\ln\left(1-\varepsilon^2 \right) +\BO{\frac{t'}{n}} \\
& \geq & \frac{t'}{2(n-t)} \ln\left(\frac{t\sqrt{v}}{n}\right) + \frac{t'}{2(n-t)}\ln(t'/t)- \frac{t' v}{4(n-t)}\ln\left(1-\varepsilon^2 \right) +\BO{\frac{t'}{n}} \\
&\geq &-  \frac{t'}{8(n-t)} \ln n - \frac{t' v}{4(n-t)}\ln\left(1-\varepsilon^2 \right) +\BO{\frac{t' \ln(t'/t)}{n}}.
\end{eqnarray*}
By plugging this expression in \eqref{eq:S0} we obtain
$$
\prob(S_0 \geq t'/2) \leq \dfrac{1}{\sqrt{t'}}
e^{ \frac{t' v}{4} \ln\left(1-\varepsilon^2 \right) + \frac{t'}{8} \ln\left(n \right) + \BO{t' \ln(t'/t)}}
$$
On the other hand we have

\begin{eqnarray}
 \prob(S_1 \geq t'/2) & \leq & 
 \dfrac{(1-q_1)\sqrt{\frac{t' }{2t}}}{(\frac{t' }{2t}-q_1)\sqrt{2\pi t (1-\frac{t' }{2t})}}e^{-t\Kull{\frac{t' }{2t}}{q_1}} 
\nonumber \\
 & \leq & 
\BO{ \dfrac{1}{\sqrt{t'}} e^{-t \Kull{\frac{t'}{2t}}{q_1}}} \label{eq:S1}
\end{eqnarray}
Similarly to what we did above,  by using the upper-bound on $q_1$ of 
Lemma \ref{lem:qb} and $\Kull{x}{y} \geq \Kull{x}{y'}$ for 
$0<y<y'<x<1$, we deduce that
$$
\Kull{\frac{t' }{2t}}{q_1} \geq  \Kull{\frac{t' }{2t}}{\BO{\frac{ (1-\varepsilon^2)^{v/2}}{ \varepsilon \sqrt{v}}}}
$$
By using this together with Lemma \ref{lem:DK} we obtain
\begin{eqnarray*}
\Kull{\frac{t'}{2t}}{q_1} & \geq & -h(t'/2t) - \frac{t'}{2t} \ln \left( \BO{\frac{ (1-\varepsilon^2)^{v/2}}{ \varepsilon \sqrt{v}}} \right) + \BO{\frac{(1-4\varepsilon^2)^{v/2}}{ \varepsilon \sqrt{v}}}\\
&\geq & - \frac{t'v}{4}  \ln\left(1-\varepsilon^2 \right) + \frac{t'}{8t} \ln n + \BO{\frac{t'}{t} \ln(t'/t)}.
\end{eqnarray*}
By using this lower-bound in \eqref{eq:S1}, we deduce
\begin{eqnarray*}
\prob(S_1 \geq t'/2) & \leq & \dfrac{1}{\sqrt{t'}} e^{\frac{t' v}{4} \ln\left(1-\varepsilon^2 \right) + \frac{t'}{8} \ln\left(n \right) + \BO{t'\ln(t'/t)}}.
\end{eqnarray*}
\end{proof}